\documentclass[english,aps,onecolumn,preprint]{revtex4}
\usepackage[T1]{fontenc}
\usepackage[latin1]{inputenc}
\usepackage{amsmath}
\usepackage{amssymb}

\makeatletter

\makeatletter

\makeatletter
\usepackage{amsfonts}

\providecommand{\U}[1]{\protect \rule{.1in}{.1in}}
\newtheorem{theorem}{Theorem}

\newtheorem{corollary}[theorem]{Corollary}

\newtheorem{lemma}[theorem]{Lemma}

\newtheorem{proposition}[theorem]{Proposition}
\newtheorem{remark}[theorem]{Remark}

\newenvironment{proof}[1][Proof]{\noindent \textbf{#1.} }{\  \rule{0.5em}{0.5em}}

\makeatother

\makeatother

\usepackage{babel}
\makeatother

\begin{document}

\title{On the distribution of linear combinations of eigenvalues of the
Anderson model}

\author{Shmuel Fishman}

\email{fishman@physics.technion.ac.il}

\author{Yevgeny Krivolapov}

\email{evgkr@tx.technion.ac.il}

\affiliation{Physics Department, Technion - Israel Institute of Technology, Haifa
32000, Israel.}

\author{Avy Soffer}

\email{soffer@math.rutgers.edu}

\affiliation{Mathematics Department, Rutgers University, New-Brunswick, NJ 08903,
USA.}


\begin{abstract}
Probabilistic estimates on linear combinations of eigenvalues of the
one dimensional Anderson model are derived. So far only estimates
on the density of eigenvalues and of pairs were found by Wegner and
by Minami. Our work was motivated by perturbative explorations of
the Nonlinear Schrödinger Equation, where linear combinations of eigenvalues
are the denominators and evaluation of their smallness is crucial.
\end{abstract}

\date{August 29, 2008}

\maketitle

\section{Introduction}

The study of the statistical properties of eigenvalues and their distribution
for the Anderson model has attracted a lot of attention in the Mathematical
Physics and Physics communities. The one dimensional Anderson model
on a lattice can be defined as an eigenvalue problem\begin{equation}
E\psi_{E}\left(x\right)=\psi_{E}\left(x-1\right)+\psi_{E}\left(x+1\right)+\varepsilon_{x}\psi_{E}\left(x\right)\qquad x\in\mathbb{Z}\label{eq:Anderson_model}\end{equation}
where $\psi_{E}\left(x\right)$ is the eigenfunction corresponding
to the eigenvalue $E$. The sites of the lattice are denoted by $x$
while $\varepsilon_{x}$ are independed random variables which are
unifomly distributed in the interval $\left[-\Delta,\Delta\right]$.
It exhibits compelete localization and does not have any exceptional
states. Important progress in the understanding of the distribution
of eigenvalues has been made at the rigorous level with the fundamental
works of Wegner \cite{Wegner} (also see e.g. Aizenman-Molchanov
(AM) \cite{AM}) and Minami \cite{Minami,GV}, which describes the
main aspects of the distribution of the distances between the eigenvalues.
Wegner's estimate shows that the typical distance is of order $1/V$,
where V is the volume of the cutoff Hamiltonian. A lot of numerical
and other arguments, led to the conjecture that the distribution of
the distances between eigenvalues (EV) comes as from a Poissonian
distribution, and the probability to have two EVs in an interval of
size $I$ is proportional to $I^{2}$. This last statement was finally
proved by Minami \cite{Minami}. Very recently, we showed that the
distribution of the energy level distances, and higher order combinations
are crucial to understanding of the large time behavior of the Nonlinear
Schrödinger Equation (NLSE) with the Anderson potential in its linear
part \cite{FKS1}. This can be understood when one studies the appropriate
time dependent perturbative construction of the solution, around the
linear part. Then, the leading potentially large terms consist of
denominators which are linear combinations (with integer coefficients)
of the various EVs of the linear or renormalized linear problem. Thus
we are naturally led to consider the probabilistic estimates on the
smallness of quantities like \begin{equation}
f=\sum_{k=1}^{R}c_{k}E_{i_{k}}\label{eq:en_sum}\end{equation}
 where the sum over $k$ is finite. Here $c_{k}$ are integers and
$E_{i}$ are the eigenvalues of the Anderson model (\ref{eq:Anderson_model}).
The indexing of the EV we use is such that to $E_{i}$ corresponds
the eigenfunction which is localized around the lattice site $i$
\cite{CGK2,GK}.

Very recently progress on this problem has been achieved in the works
\cite{AW} and \cite{CGK}. Our aim in this work is to obtain the
needed estimates for the perturbative construction of large time solutions
for the NLSE with Anderson potentials \cite{WZ} , as developed in
\cite{FKS1,FKS2}. Our strategy is to begin with the AM approach
\cite{AM}, and estimate the expectations of fractional powers of
the inverse of the linear combinations of EVs, denoted by $f$ (Sec.
III). It turns out that the critical step in this problem is to estimate
from above the probability of avoided crossings, which is equivalent
to controlling the probability of states, with two humps that are
far away. Another useful tool is the demonstration of a \textbf{lower}
bound on the eigenfunctions, which is uniform in the potential, with
large probability (Sec. II). A modified version of $f$ is considered
in Sec. IV.

\section{The asymptotic exponential decay of the eigenstates}

For each energy $E$ the inverse localization length, $\gamma\left(E\right),$
is continuous as a result of the Thouless formula \cite{Thouless}
and, $\gamma\left(E\right)>0$ (see more details after (\ref{eq:en_cover})).
Since the spectrum of a Hamiltonian $H$ on a lattice is compact,
it follows that, given $\bar{\eta},$ we can cover it with a finite
collection of open intervals, such that each interval is around $\lambda_{k},\left.k=1,2...K<\infty,\right.$
and for energies in that interval $\left.\left(1-\bar{\eta}\right)\gamma\left(\lambda_{k}\right)\leq\gamma\left(E\right)\leq\left(1+\bar{\eta}\right)\gamma\left(\lambda_{k}\right)\right.,$
where $\bar{\eta}$ is arbitrarily small.

Now we focus on one of such intervals located around the energy $\lambda_{0}\in I_{0},$
and fixed $\left.\varepsilon\ll2\bar{\eta}\gamma\left(\lambda_{0}\right)\right..$
We denote by $s_{n,j}$ the set of all the potentials for which the
energy $E_{j}\in I_{0},$ and such that, \begin{equation}
\left\vert e^{\left(\gamma\left(E_{j}\right)+\varepsilon\right)\left\vert x\right\vert }\psi_{E_{j}}\left(x\right)\right\vert \geq1\qquad j\in\mathbb{Z}\label{eq:cond}\end{equation}
 the inequality holds for all $\left\vert x\right\vert \geq n\left(\bar{\eta},\varepsilon\right),$
and $x$ is the distance from the localization center of the wave
function. It is a consequence of the Furstenberg theorem \cite{Furst}
and of the analysis of \cite{Souillard} that $n$ is finite for
a.e. potential (for an alternative approach see \cite{Klein}).

The measure of the potentials belonging to $s_{n,j}$ is denoted by
$S_{n,j}.$ Clearly, $S_{n,j}\geq0,$ and we only consider the case
where the measure of the set of potentials with $E_{j}$ belonging
to the interval around $\lambda_{0}$ is nonzero. $S_{n,j}$ is a
monotonic increasing sequence of $n$: this follows since by construction,
if $S_{n,j}$ is such that condition (\ref{eq:cond}) is satisfied
for $\left\vert x\right\vert \geq n_{j}$ it is also satisfied for
$\left\vert x\right\vert \geq n_{j}+l,~l>0.$ Furthermore, one can
verify that each $S_{n,j}$ is measurable.

\begin{lemma} The limit $\lim_{n\rightarrow\infty}S_{n,j}=\mu'_{0}>0,$
exists with some positive $\mu'_{0},$ where $\mu'_{0}$ is the measure
of all the potentials, such that $E_{j}$ is in the interval around
$\lambda_{0}.$ \end{lemma}

\begin{proof} The lemma follows, since by Furstenberg theorem \cite{Furst}
and Delyon-Souillard-Levy \cite{Souillard} for almost all energies
the eigenfunctions decay exponentially with the rate $\gamma\left(E\right):$\begin{equation}
\lim_{x\rightarrow\infty}\frac{\ln\left[\left\vert \psi_{E_{j}}\left(x\right)\right\vert ^{2}+\left\vert \psi_{E_{j}}\left(x+1\right)\right\vert ^{2}\right]}{2x}=\gamma\left(E_{j}\right).\end{equation}
 Therefore\begin{equation}
\left(\left\vert \psi_{E_{j}}\left(x\right)\right\vert ^{2}+\left\vert \psi_{E_{j}}\left(x+1\right)\right\vert ^{2}\right)^{1/2}\asymp P\left(x\right)e^{-\gamma\left(E_{j}\right)\left\vert x\right\vert }\end{equation}
 with $P\left(x\right)$ sub-exponentially bounded, above and below
(decaying or increasing), and $x$ is the distance from the localization
center \cite{Rio,KS}.

Hence\begin{equation}
\lim_{x\rightarrow\infty}\left\vert e^{\left(\gamma\left(E_{j}\right)+\varepsilon\right)\left\vert x\right\vert }P\left(x\right)e^{-\gamma\left(E_{j}\right)\left\vert x\right\vert }\right\vert =\infty\end{equation}
 for all $\varepsilon>0,$ and therefore for each eigenfunction condition
(\ref{eq:cond}) is satisfied for some \emph{finite} $n_{j}.$ Hence
all realizations of the potential satisfy condition (\ref{eq:cond})
for any $E_{0}$ in the interval around $\lambda_{0}:$\begin{equation}
\mu_{0}=\mu\left(\text{all potentials}\right)=\mu\left({\displaystyle \bigcup_{n,j}}s_{n,j}\right)\leq{\displaystyle \sum\limits _{n,j}}\mu\left(s_{n,j}\right),\end{equation}
 since $\left\{ n,j\right\} $ is countable.\end{proof}

The lemma holds for each of the open sets around $\lambda_{k}.$ Since
the number of these sets is finite, then for any fixed $n_{j}$ the
condition (\ref{eq:cond}) is satisfied for a set of realizations
with a measure of $1-\delta_{j},$ where $\delta_{j}$ decreases with
$n_{j}.$ 

This leads to the following proposition.

\begin{proposition} \label{prop:asymp} (Lower bound on Eigenfunctions)
Given $\varepsilon,\delta>0,$ with probability $1-\delta,$ $\exists\left\{ n_{j}^{\ast}\left(\varepsilon,\delta\right)\right\} \subseteq\mathbb{N},$
such that\begin{equation}
\left(\left\vert \psi_{E_{j}}\left(n\right)\right\vert ^{2}+\left\vert \psi_{E_{j}}\left(n+1\right)\right\vert ^{2}\right)^{1/2}\asymp P\left(n\right)e^{-\gamma\left(E_{j}\right)\left\vert n\right\vert }\geq e^{-\left(\gamma\left(E_{j}\right)+\varepsilon\right)\left\vert n\right\vert }\end{equation}
 for all $n\geq n_{j}^{\ast}\left(\varepsilon,\delta\right)$, and
where $n$ is distance from the localization center. The proof of
the upper bound is similar. \end{proposition}

\begin{corollary} \label{lemma:deriv}For any $f={\displaystyle \sum\limits _{k=1}^{R}}c_{k}E_{i_{k}},$
there is an $x$ such that for $j>\left\vert x\right\vert ,$\begin{equation}
\Pr\left(\left\vert \frac{\partial f}{\partial\varepsilon_{j}}+\frac{\partial f}{\partial\varepsilon_{j+1}}\right\vert \leq C_{j}\right)\leq\bar{\delta}\left(j\right)\end{equation}
 where $C_{j}>e^{-2(\tilde{\gamma}+\varepsilon)j}$ and $\bar{\delta}\left(j\right)$
is monotonically decreasing with $j$. The decay rate $\tilde{\gamma}$
corresponds to an energy $E_{i_{k}}$ in the sum for $f$, namely
$\tilde{\gamma}=\gamma(E_{i_{k}})$. \end{corollary}

\begin{proof} Using Feynman-Hellman theorem\begin{equation}
\left\vert \frac{\partial f}{\partial\varepsilon_{j}}+\frac{\partial f}{\partial\varepsilon_{j+1}}\right\vert =\left\vert {\displaystyle \sum\limits _{k=1}^{R}}c_{k}\left(\left\vert \psi_{i_{k}}\left(j\right)\right\vert ^{2}+\left\vert \psi_{i_{k}}\left(j+1\right)\right\vert ^{2}\right)\right\vert .\label{eq:deriv_sum}\end{equation}

First we order the eigenfunctions in the RHS of (\ref{eq:deriv_sum})
by their localization centers, $\left\{ i_{1}<i_{2}<\cdots<i_{R}\right\} ,$
where $R$ is the number of eigenfunctions in the sum. If the eigenfunctions
are tightly packed in the sense, that their localization centers are
found in a box of size $n^{\ast}\left(\varepsilon,\delta\right)$
then we can take $j$ such that $\left.\left\vert j-i_{R}\right\vert >n^{\ast}\left(\varepsilon,\delta\right)\right.$
and Proposition \ref{prop:asymp} would apply for all the eigenfunctions
in the sum. We assume for simplicity $j>0$. Here $n^{*}\left(\varepsilon,\delta\right)\equiv\sup_{l\epsilon\left\{ i_{m}\right\} _{m=1}^{R}}n_{l}^{*}\left(\varepsilon\right).$
For $j$ large enough the asymptotic behavior of this sum is dictated
by one term with the smallest inverse localization length, $\tilde{\gamma}$.
In the worst case, there are two states which are localized at a distance
of $n^{\ast}$ apart on the lattice and have very close inverse localization
lengths $\tilde{\gamma},\left(\tilde{\gamma}\pm\eta\right)$, with
$\eta$ small. By continuity of $\gamma\left(E\right)$ the respective
energies are typically close. For these states we have to go for a
distance of\begin{equation}
j>\gamma_{\min}n^{\ast}/\eta,\label{eq:j}\end{equation}
 where $\gamma_{\min}=\inf_{E}\gamma\left(E\right).$ This results
from the requirement that $\left.e^{-\gamma_{\min}j}\gg e^{-\left(\gamma_{\min}+\eta\right)\left(j-n^{\ast}\right)}\right.,$
which is satisfied for $\left.e^{\left(\gamma_{\min}+\eta\right)n^{\ast}-\eta j}\ll1\right..$
We will disregard realizations of the potentials for which many $\gamma^{\prime}s$
are closer than $\eta,$ however we notice that the measure of such
realizations goes to zero with $\eta.$ In fact the smallest energy
spacing results from avoided crossings which are exponentially small
in $n^{\ast}.$ In this process we practically eliminated realizations
where arbitrarily small energy differences are found. (See next section
for detailed discussion of this issue, a quantitative estimate is
(\ref{m17})). For a general configuration of the localization centers
in the sum we will separate the localization centers to clusters such
that the distance between two adjacent clusters is at least $n^{*}\gamma_{\min}/\eta$.
Following from above, there are intervals on the lattice between the
clusters, such that the contribution of each cluster to the sum is
governed by one exponential. Therefore we can always find $j^{\prime}s$
between the clusters such that only one cluster and actually one state
dominates, moreover at this point the magnitude of the sum satisfies
the inequality\begin{equation}
C_{j}\geq e^{-\gamma_{\max}\gamma_{\min}n^{\ast}/\eta},\end{equation}
 where $\gamma_{\max}=\sup_{E}\gamma\left(E\right).$ In this way
we obtain the following proposition.

\begin{proposition} Given $R$ energies $E_{i_{1}},E_{i_{2}},...,E_{i_{R}}$
we can always decompose this set into a set of disjoint clusters such
that the largest $n^{*}$ is $cR/\gamma_{min}$, where with probability
$1-\varepsilon,for$ $j>n^{*}\left(\varepsilon,R\right)$ \[
\sum_{k=1}^{R}c_{k}\left|\psi_{E_{i_{k}}}\left(j\right)\right|^{2}\geq ce^{-2\left(\tilde{\gamma}+\varepsilon\right)\left|j\right|},\]
 with $c_{k}$ that are integers and $c$ is a constant. That is for
some $j$, some $\left|\psi_{E_{i_{k}}}\left(j\right)\right|^{2}$
dominates all other $\left|\psi_{E_{i_{k}}}(j)\right|^{2}$.\end{proposition}

Using the above proposition we can find an $x,$ such that for $j>x,$
the sum (\ref{eq:deriv_sum}) will be dominated by one term with the
smallest $\gamma\left(E\right),$ such that\begin{equation}
\Pr\left(\left\vert \frac{\partial f}{\partial\varepsilon_{j}}\right\vert \leq C_{j}\right)\leq\bar{\delta}\left(j\right),\end{equation}
 and $C_{j}$ depends only on the relative distance between $j$ and
the localization center of the dominant eigenfunction, and therefore
does not depend on the locations of the other localization centers
of the states in $f$. Here $\bar{\delta}\left(j\right)$ is the measure
of all the realizations that were excluded so that Proposition \ref{prop:asymp}
holds. Due to the Proposition \ref{prop:asymp} both $C_{j}$ and
$\bar{\delta}\left(j\right),$ confined with (\ref{eq:j}) are decreasing
functions of $j.$ Therefore $C_{j}$ and $\bar{\delta}\left(j\right)$
are bounded by $C_{x}$ and $\bar{\delta}\left(x\right)$ \end{proof}

\section{The integral over $\left\vert f\right\vert ^{-s}$}

\begin{lemma} \label{lemma:crossings}For each $j,$ exists a finite
interval $\Delta\varepsilon_{j}>0$ for which the function $\frac{\partial f}{\partial\varepsilon_{j}}+\frac{\partial f}{\partial\varepsilon_{j+1}}$
does not change sign with probability $1-\bar{\delta}_{1}(j)$ \end{lemma}

\begin{proof} The function $\frac{\partial f}{\partial\varepsilon_{j}}+\frac{\partial f}{\partial\varepsilon_{j+1}}$
is given by\begin{equation}
\frac{\partial f}{\partial\varepsilon_{j}}+\frac{\partial f}{\partial\varepsilon_{j+1}}={\displaystyle \sum\limits _{k=1}^{R}}c_{k}\left(\left\vert \psi_{E_{i_{k}}}\left(j\right)\right\vert ^{2}+\left\vert \psi_{E_{i_{k}}}\left(j+1\right)\right\vert ^{2}\right).\end{equation}
 Due to Corollary \ref{lemma:deriv} there is a $j\geq j_{0}$ for
which the sum is dominated by only one term with the smallest $\gamma\left(E\right).$
We will denote the state corresponding to this energy by its eigenvalue
$E_{j}^{0}$. The sign of $\frac{\partial f}{\partial\varepsilon_{j}}+\frac{\partial f}{\partial\varepsilon_{j+1}}$
can change only when there is an avoided crossing between the $E_{j}^{0}$
state and one of the other states $E_{k}$, in the sum defining $f$,
since then the dominating terms will interchange due to the fact that
$\gamma\left(E\right)$ is a continuous function of $E$. Consider
a cluster of size $L$, which is large enough to cover all the eigenstates,
such that outside of this cluster the relevant eigenstates are smaller
than $e^{-\gamma_{\min}n^{*}}$. Using this cluster, the Minami estimate
\cite{Minami} gives\begin{equation}
\Pr\left(TrP_{H_{\omega}}^{(L)}(I)\geq2\right)\leq\left(\pi\left\Vert \rho\right\Vert _{\infty}IL\right)^{2},\end{equation}
 where $\left\Vert \rho\right\Vert _{\infty}$ is the supremum of
the density of states, $I$ is some energy interval while $P_{H_{\omega}}^{(L)}(I)$
is the spectral projection on that interval. We can reformulate this
estimate in such a manner that the probability to find any $i$ and
$k$ with energies in this interval satisfies \begin{equation}
\Pr\left(\left\vert E_{i}-E_{k}\right\vert \leq I\right)\leq\left(\pi\left\Vert \rho\right\Vert _{\infty}IL\right)^{2}.\end{equation}
 Dividing the spectrum into intervals of length $I=\delta_{1}L^{-\nu}$
with $\delta_{1}>0$ and $\nu>2$, we find that the probability to
find any pair $i,k$ in such an interval satisfies \begin{equation}
\Pr\left(\left\vert E_{i}-E_{k}\right\vert \leq\delta_{1}L^{-\nu}\right)\leq\delta_{1}\Delta\pi^{2}\left\Vert \rho\right\Vert _{\infty}^{2}L^{-\left(2\nu-2\right)}.\label{eq:en_cover}\end{equation}
 Therefore the probability of an energy interchange can be made arbitrarily
small. After the interchange of the energies the magnitudes of the
corresponding $\gamma^{\prime}s$ will interchange and one (different)
state will dominate at $j$. This assumes that $d\gamma/dE$ for this
state is not too small. If it vanishes the following evaluation is
required. In this situation, the change in $\gamma$ for a small change
$\delta_{3}$ in $E,$ we assume that there is a finite number of
points, $N_{b}$, where the derivative of $\gamma(E)$ is zero, which
follows from the analyticity of $\gamma\left(E\right)$, that in turn
follows from the Thouless formula \cite{Thouless} and the analyticity
of the density of states \cite{HKS,S,Taylor}. We take a small interval,
$\delta_{2}$, around those points for which, $d\gamma/dE\leq\delta_{2}$.
We will exclude realizations with energies which are found within
this interval, by using the Minami estimate, \begin{equation}
\Pr\left(\left\vert E_{i}-E_{k}\right\vert \leq I(\delta_{2})\right)\leq N_{b}\left(\pi\left\Vert \rho\right\Vert _{\infty}I(\delta_{2})L\right)^{2}\leq\tilde{\delta}_{1}(j).\end{equation}
 where $I(\delta_{2})$ is the interval of energies for which , $d\gamma/dE\leq\delta_{2}$.
This interval, $I(\delta_{2})$, is a monotonically decreasing function
of $\delta_{2}$. Energies which are found outside of this interval
will have a derivative $d\gamma/dE\geq\delta_{2}$ and therefore using
(\ref{eq:en_cover}) \begin{equation}
\Pr\left(\left\vert \gamma(E_{i}-\gamma(E_{j})\right\vert \leq\delta_{1}\delta_{2}L^{-\nu}\right)\leq\pi^{2}\delta_{1}\Delta\left\Vert \rho\right\Vert _{\infty}^{2}L^{-\left(2\nu-2\right)}\leq\tilde{\delta}_{2}(j).\label{m17}\end{equation}
 where $\tilde{\delta}_{2}(j)$ can be made arbitrarily small. Here
we choose $L$ to be much larger than the cluster in question. To
ensure that this probability is small we have to choose $\nu>2$.
Note that the number of intervals grows as $L^{\nu}$. The realizations
which are excluded are of probability which is bounded by $\bar{\delta}_{1}(j)=\tilde{\delta}_{1}(j)+\tilde{\delta}_{2}(j)$.
\end{proof}

\begin{remark} \label{remark:1} The bound on the probability of
the previous corollary should be modified with $\bar{\delta}_{j}$
replaced by $\bar{\delta}(j)+\bar{\delta}_{1}(j)$. \end{remark}

\begin{theorem} \label{th:main_theorem}For $f={\displaystyle \sum\limits _{k=1}^{R}}c_{k}E_{i_{k}},$the
following mean is bounded from above\begin{equation}
\left\langle \frac{1}{\left\vert f\right\vert ^{s}}\right\rangle _{\delta}={\displaystyle \int}{\displaystyle \prod\limits _{i=1}^{\left\vert \Lambda\right\vert }}d\mu\left(\varepsilon_{i}\right)\frac{1}{\left\vert f\right\vert ^{s}}\leq D_{\delta}.\end{equation}
 where $\left\langle .\right\rangle _{\delta}$ denotes a mean over
realizations of a measure $1-\delta$ and $D_{\delta}$ is independent
of $\left\vert \Lambda\right\vert .$We analyze the behavior of the
integral\begin{equation}
I_{f}\equiv{\displaystyle \int}{\displaystyle \prod\limits _{i=1}^{\left\vert \Lambda\right\vert }}d\mu\left(\varepsilon_{i}\right)\frac{1}{\left\vert f\left(\vec{\varepsilon}\right)\right\vert ^{s}}\end{equation}

\end{theorem}

\begin{proof} We denote, $\vec{\varepsilon}\equiv\left\{ \varepsilon_{1,}\cdots\varepsilon_{\left\vert \Lambda\right\vert }\right\} ,$
and assume, that $\mu\left(\varepsilon_{i}\right)$ is uniform in
the interval $\left[-\Delta,\Delta\right]$. We first change variables
(for a $j$ chosen to optimize the bound on (\ref{f20})) to \begin{eqnarray*}
\varepsilon_{j}^{+} & = & \frac{1}{\sqrt{2}}\left(\varepsilon_{j}+\varepsilon_{j+1}\right)\\
\varepsilon_{j}^{-} & = & \frac{1}{\sqrt{2}}\left(\varepsilon_{j}-\varepsilon_{j+1}\right)\end{eqnarray*}
 with the Jacobian equal to $1$. Then we change the variables to
$\left(f,\varepsilon_{1},\ldots,\varepsilon_{j}^{-},\varepsilon_{j+1},\ldots\varepsilon_{\left\vert \Lambda\right\vert }\right),$
that produces the Jacobian\begin{equation}
\frac{\partial\left(f,\varepsilon_{1},\ldots,\varepsilon_{\left\vert \Lambda\right\vert }\right)}{\partial\left(\varepsilon_{1},\ldots,\varepsilon_{\left\vert \Lambda\right\vert }\right)}=\left\vert \begin{array}{cccc}
\frac{\partial f}{\partial\varepsilon_{j}^{+}} & \frac{\partial f}{\partial\varepsilon_{1}} & \cdots & \frac{\partial f}{\partial\varepsilon_{\left\vert \Lambda\right\vert }}\\
 & 1\\
 &  & \ddots\\
 &  &  & 1\end{array}\right\vert =\left\vert \frac{\partial f}{\partial\varepsilon_{j}^{+}}\right\vert .\label{f20}\end{equation}
 The integral using the new variables is\begin{equation}
I_{f}=\frac{1}{\left(2\Delta\right)^{\left\vert \Lambda\right\vert }}\int d\varepsilon_{j}^{-}{\displaystyle \int_{-\Delta}^{\Delta}}{\displaystyle \prod\limits _{i=1,i\neq j}^{\left\vert \Lambda\right\vert }}d\varepsilon_{i}\int_{f\left(\vec{\varepsilon},-\overrightarrow{\Delta}\right)}^{f\left(\vec{\varepsilon},\vec{\Delta}\right)}\left\vert \frac{\partial f}{\partial\varepsilon_{j}^{+}}\right\vert ^{-1}df\frac{1}{\left\vert f\right\vert ^{s}}.\end{equation}
 In order to take into the account the multiplicity, which results
of this change of variables, we will divide the range of integration
in the $\vec{\varepsilon}$ space into domains $m_{l},$ $l=1,...,M$
, such that in each domain the sign of the Jacobian does not change.
This results in\begin{equation}
I_{f}=\frac{1}{\left(2\Delta\right)^{\left\vert \Lambda\right\vert }}\sum_{l=1}^{M}{\displaystyle \int_{m_{l}}}{\displaystyle \prod\limits _{i=1,i\neq j}^{\left\vert \Lambda\right\vert }}d\varepsilon_{i}\int_{f_{l}^{-}}^{f_{l}^{+}}\left\vert J_{l}\right\vert ~df\frac{1}{\left\vert f\right\vert ^{s}}\end{equation}
 where $J_{l}$ is the Jacobian $\left\vert \frac{\partial f}{\partial\varepsilon_{j}^{+}}\right\vert ^{-1}$
in the $l-th$ domain and $f_{l}^{\pm}$ are the extremal values of
$f$ in this domain. Such splitting into a finite number of domains
is possible due to Lemma \ref{lemma:crossings} and since $\left\vert \frac{\partial f}{\partial\varepsilon_{j}^{+}}\right\vert =\frac{1}{\sqrt{2}}\left\vert \frac{\partial f}{\partial\varepsilon_{j}}+\frac{\partial f}{\partial\varepsilon_{j+1}}\right\vert $.
Since the integrand is positive the value of the integral can be only
increased by increasing the integration volume. Therefore we obtain
the following inequality \begin{equation}
\left\vert f\right\vert =\left\vert {\displaystyle \sum\limits _{k=1}^{R}}c_{k}E_{i_{k}}\right\vert \leq\max_{\omega}\left\vert E_{i_{k}}\right\vert {\displaystyle \sum\limits _{k}}\left\vert c_{k}\right\vert =\left(2+\Delta\right){\displaystyle \sum\limits _{k}}\left\vert c_{k}\right\vert \leq Q\end{equation}
 where we have used the fact that the eigenvalues are bounded and
the sum is finite. The integral can be bounded by \begin{align}
I_{f} & =\frac{1}{\left(2\Delta\right)^{\left\vert \Lambda\right\vert }}\sum_{l=1}^{M}{\displaystyle \int_{m_{l}}^{}}{\displaystyle \prod\limits _{i=1,i\neq j}^{\left\vert \Lambda\right\vert }}d\varepsilon_{i}\int_{f_{l}^{-}}^{f_{l}^{+}}\left\vert J_{l}\right\vert ~df\frac{1}{\left\vert f\right\vert ^{s}}\\
 & \leq\int_{0}^{Q}df\frac{1}{f^{s}}\sup_{\left\{ \varepsilon_{i}\right\} }\left(\left\vert \frac{\partial f}{\partial\varepsilon_{j}^{+}}\right\vert ^{-1}\right)\left(\frac{1}{\left(2\Delta\right)^{\left\vert \Lambda\right\vert }}\sum_{l=1}^{M}{\displaystyle \int_{m_{l}}^{}}{\displaystyle \prod\limits _{i=1,i\neq j}^{\left\vert \Lambda\right\vert }}d\varepsilon_{i}\right)\nonumber \end{align}
 with $\left\vert J_{l}\right\vert \leq\sup_{\left\{ \varepsilon_{i}\right\} }\left(\left\vert \frac{\partial f}{\partial\varepsilon_{j}^{+}}\right\vert ^{-1}\right).$
The last term is an integral over all $\varepsilon_{j}$ excluding
one, therefore the volume transforms as \begin{equation}
\left(\frac{1}{\left(2\Delta\right)^{\left\vert \Lambda\right\vert }}\sum_{l=1}^{M}{\displaystyle \int_{m_{l}}^{}}{\displaystyle \prod\limits _{i=1,i\neq j}^{\left\vert \Lambda\right\vert }}d\varepsilon_{i}\right)=\frac{1}{2\Delta}.\end{equation}
 Using Corollary \ref{lemma:deriv} and Remark \ref{remark:1} the
integral can be further bounded by\begin{equation}
I_{f}\leq C_{\delta}^{-1}\frac{Q^{1-s}}{2\Delta\left(1-s\right)}\equiv D_{\delta}\label{eq:If}\end{equation}
 with $D_{\delta}$ independent of of the size of system and of $R$
while $C_{\delta}$ is the smallest $C_{j}.$ \end{proof}

\section{The effect of renormalization}

In some applications \cite{FKS1,FKS2} the function $f=\sum_{k}c_{k}E_{i_{k}}$
should be replaced by $f^{\prime}=\sum_{k}c_{k}E_{i_{k}}^{\prime}$,
where $E_{i}^{\prime}$ are the renormalized energies and $f^{\prime}=f+\beta F(\left\{ \overrightarrow{\varepsilon}\right\} )$.
In this application, to the leading order only one energy denoted
by $E_{0}$ is renormalized, namely \[
E_{i}^{\prime}=E_{i}+\beta\delta_{i,0}V_{0}^{000}\]
 where $\beta$ is a small parameter and $V_{0}^{000}=\sum_{i}\psi_{E_{0}}^{4}(i)$.
We show that to this order $C_{\delta}$ is not affected by $F$ for
$\beta<$const $e^{-2\left(\gamma_{\max}-\gamma_{\min}\right)\left\vert x_{\delta}\right\vert },$
where $\gamma_{\min}$ and $\gamma_{\max}$ are the minimal and the
maximal Lyapunov exponents and $\left\vert x_{\delta}\right\vert $
is defined by $C_{\delta}=\left(C_{j},j=x_{\delta}\right).$

The derivative of $V_{0}^{000}$ is given by\begin{equation}
\frac{\partial V_{0}^{000}}{\partial\varepsilon_{j}^{}}={\displaystyle \sum\limits _{i}}4\psi_{E_{0}}^{3}\left(i\right)\frac{\partial\psi_{E_{0}}\left(i\right)}{\partial\varepsilon_{j}}.\end{equation}

\begin{lemma}\begin{equation}
\Pr\left(\left\vert \frac{\partial V_{0}^{000}}{\partial\varepsilon_{j}^{+}}\right\vert \geq De^{-2\left(\gamma_{\min}-\eta\right)\left\vert x_{j}\right\vert }\right)\leq e^{-2\eta s\left\vert x_{j}\right\vert }\qquad0<s<1\end{equation}
 where $\gamma_{\min}$ is the minimal Lyapunov exponent. \end{lemma}

\begin{proof} Using a perturbation expansion in $\varepsilon_{j}$
(to the first order), this derivative, $\frac{\partial\psi_{E_{n}}\left(i\right)}{\partial\varepsilon_{j}}$,
can be calculated\begin{equation}
\frac{\partial\psi_{E_{n}}\left(i\right)}{\partial\varepsilon_{j}}=\psi_{E_{n}}\left(j\right){\displaystyle \sum\limits _{k\neq n}}\frac{\psi_{E_{k}}\left(j\right)\psi_{E_{k}}\left(i\right)}{E_{n}-E_{k}}.\end{equation}
 We now use the fractional moment technique to get a probabilistic
bound on this derivative. Using the bound on the eigenfunctions and
the estimate on the average of $\left\langle \frac{1}{\left\vert E_{n}-E_{k}\right\vert ^{s}}\right\rangle _{\delta},$
(see Theorem \ref{th:main_theorem}), we obtain \begin{align}
\left\langle \left\vert \frac{\partial\psi_{E_{n}}\left(i\right)}{\partial\varepsilon_{j}}\right\vert ^{s}\right\rangle _{\delta} & =\left\langle \left\vert \psi_{E_{n}}\left(j\right)\right\vert ^{s}\left\vert {\displaystyle \sum\limits _{k\neq n}}\frac{\psi_{E_{k}}\left(j\right)\psi_{E_{k}}\left(i\right)}{E_{n}-E_{k}}\right\vert ^{s}\right\rangle _{\delta}\\
 & \leq\sup_{\omega}\left(\left\vert \psi_{E_{n}}\left(j\right)\right\vert ^{s}{\displaystyle \sum\limits _{k\neq n}}\left\vert \psi_{E_{k}}\left(j\right)\right\vert ^{s}\left\vert \psi_{E_{k}}\left(i\right)\right\vert ^{s}\right)\left\langle \frac{1}{\left\vert E_{n}-E_{k}\right\vert ^{s}}\right\rangle _{\delta}\nonumber \\
 & \leq\bar{D}_{\varepsilon}D_{\delta}e^{-\gamma_{\min}s\left\vert x_{j}-x_{n}\right\vert }e^{-\gamma_{\min}s\left\vert x_{j}-x_{i}\right\vert }.\nonumber \end{align}
 Therefore the derivative $\frac{\partial V_{0}^{000}}{\partial\varepsilon_{j}}$
(where $x_{0}=0$), can be bounded by\begin{align}
\left\langle \left\vert \frac{\partial V_{0}^{000}}{\partial\varepsilon_{j}}\right\vert ^{s}\right\rangle _{\delta} & \leq4\bar{D}_{\varepsilon}D_{\delta}e^{-\gamma_{\min}s\left\vert x_{j}\right\vert }{\displaystyle \sum\limits _{i}}e^{-3\gamma_{\min}s\left\vert x_{i}\right\vert }e^{-\gamma_{\min}s\left\vert x_{j}-x_{i}\right\vert }\\
 & \leq4\bar{D}_{\varepsilon}D_{\delta}e^{-2\gamma_{\min}s\left\vert x_{j}\right\vert }{\displaystyle \sum\limits _{i}}e^{-2\gamma_{\min}s\left\vert x_{i}\right\vert }=\tilde{D}_{\varepsilon}D_{\delta}e^{-2\gamma_{\min}s\left\vert x_{j}\right\vert }.\nonumber \end{align}
 Using the definition of $\varepsilon_{j}^{+}$\[
\left\langle \left\vert \frac{\partial V_{0}^{000}}{\partial\varepsilon_{j}^{+}}\right\vert ^{s}\right\rangle _{\delta}\leq\frac{1}{2^{s/2}}\left(\left\langle \left\vert \frac{\partial V_{0}^{000}}{\partial\varepsilon_{j}}\right\vert ^{s}\right\rangle _{\delta}+\left\langle \left\vert \frac{\partial V_{0}^{000}}{\partial\varepsilon_{j+1}}\right\vert ^{s}\right\rangle _{\delta}\right)\leq\tilde{D}_{\varepsilon}D_{\delta}e^{-2\gamma_{\min}s\left\vert x_{j}\right\vert }\]
 Utilizing the Chebychev inequality we complete the proof. \end{proof}

\begin{remark} Since the derivative $\frac{\partial V_{0}^{000}}{\partial\varepsilon_{j}}$
is bounded from above we can always select $\beta$ in such way that
its contribution will be smaller than the derivatives of the energies.
\end{remark}

\begin{acknowledgments}
This work was partly supported by the Israel Science Foundation (ISF),
by the US-Israel Binational Science Foundation (BSF), by the USA National
Science Foundation (NSF), by the Minerva Center of Nonlinear Physics
of Complex Systems, by the Shlomo Kaplansky academic chair and by
the Fund for promotion of research at the Technion. We had illuminating
and informative discussions with many of our colleagues. Of particular
importance were the discussion with M. Aizenman, I. Goldshield, M.
Golsdtein, W.-M. Wang and S. Warzel.
\end{acknowledgments}

\end{document}